\documentclass[12pt]{amsart}

\usepackage{a4wide}
\usepackage[dvips]{epsfig}
\usepackage{amscd}
\usepackage{amssymb}
\usepackage{amsthm}
\usepackage{amsfonts}
\usepackage{amsmath}
\usepackage{latexsym}
\usepackage{graphicx}

\usepackage{subfig}
\usepackage{float}

\usepackage{hyperref}
\hypersetup{
		colorlinks,
    citecolor=black,%
    filecolor=black,%
    linkcolor=black,%
    urlcolor=black,
    pdfstartview={FitH},
    pdftitle={Optimal portfolio problems in the presence of default and the value of information},
    pdfauthor={Steffen Sjursen},
}

\newtheorem{theorem}{Theorem}
\theoremstyle{plain}

\newtheorem{definition}[theorem]{Definition}
\newtheorem{example}[theorem]{Example}

\newtheorem{lemma}[theorem]{Lemma}

\newtheorem{remark}[theorem]{Remark}

\numberwithin{equation}{section}
\numberwithin{theorem}{section}

\newcommand{\E}{\ensuremath{\mathbb{E}}}

\allowdisplaybreaks

\newcommand{\Prob}{\ensuremath{\mathbb{P}}}
\newcommand{\invF}[0]{\mathcal{G}}
\newcommand{\invFF}{\mathbb{G}}
\newcommand{\Rr}{{\mathbb{R}_0}}
\newcommand{\aui}{$A_{\text{u.i.}}$}
\newcommand{\uniproc}{\Psi}
\newcommand{\ie}{i.e.~}
\newcommand{\ind}{\mathbf{1}}

\begin{document}
\title[Optimal investments in defaultable assets]{Information and optimal investment \\in defaultable assets}
\date{\today }
\author[Di Nunno]{Giulia Di Nunno}
\address{Giulia Di Nunno: Center of Mathematics for Applications, University of Oslo,
PO Box 1053 Blindern, N-0316 Oslo, Norway, and, Norwegian School of Economics and Business Administration, Helleveien 30, N-5045 Bergen, Norway.}
\author[Sjursen]{Steffen Sjursen}
\address{Steffen Sjursen: Center of Mathematics for Applications, University of Oslo, PO Box 1053 Blindern, N-0316 Oslo, Norway}
\email[]{giulian\@@math.uio.no, s.a.sjursen\@@cma.uio.no}
\urladdr{http://folk.uio.no/giulian/}
\thanks{The research leading to these results has received funding from the
European Research Council under the European Community's Seventh Framework
Programme (FP7/2007-2013) / ERC grant agreement no [228087]}





\begin{abstract}
We study optimal investment in an asset subject to risk of default for investors that rely on different levels of information. The price dynamics can include noises both from a Wiener process and a Poisson random measure with infinite activity. The default events are modeled via a counting process in line with large part of the literature in credit risk. In order to deal with both cases of inside and partial information we consider the framework of the anticipating calculus of forward integration. This does not require a priori assumptions typical of the framework of enlargement of filtrations. We find necessary and sufficent conditions for the existence of a locally maximizing portfolio of the expected utility at terminal time. We consider a large class of utility functions. In addition we show that the existence of the solution implies the semi-martingale property of the noises driving the stock. Some discussion on unicity of the maxima is included.
\end{abstract}

\maketitle

\section{Introduction: The model, the optimization problem, the streams of information}

Occasionally, we observe that unexpected events wipe out shareholder values. We will generically call all these events \emph{default events}.
Inspired by default risk literature, we consider a model for stocks where there is a varying risk of instantaneous loss in the stock value. 

Of particular interest here is when the default events are dependent on the noises driving the stock or when the investor has insider information. In these cases mathematical questions arise as to whether the driving noises are still (semi)-martingales and the relevant stochastic integrals can be interpreted in the It\^o sense. Since this is not a priori certain, we choose to investigate this issue using forward integration in the modeling of stock prices. With this we do not need a priori assumptions or restrictions on the information available to the investor and we will be able to use a unique framework for all the situations of interest.

Our main result is a sufficent and necessary criteria for an optimal investment strategy maximizing the expected utility of the final portfolio value, for a portfolio involving the defaultable asset. 
We remark that this result also holds for optimization problems with partial or delayed information. 

Furthermore we show that the existence of an optimal strategy yields the semi-martingale property of the noises. This would usually be assumed a priori if working in the framework of enlargement of filtrations see for instance \cite{Bielecki2004,coculescu2012, DeLong2006,ElKaroui2010,Jeanblanc2009,Jeanblanc2010}.

\medskip

The defaultable stock is modeled with three random noises, a Wiener process $W$, a Poisson random measure $N$ and a pure jump process $H$. The occurence of defaults or catastrophic events is modelled by $H$. The intensity of $H$, as viewed by the investor, is stochastic and can either depend on current and future knowledge of $W$ and $N$ or be independent of the two. 

Our model market on the time horizon $[0,T]$ ($T>0$) consists of a (non-defaultable) bond $S_0$ serving as num\'eraire with dynamics:
\begin{align}\label{eq:S_0}
dS_0 (t) =&\; S_0(t) \rho(t) dt,   \\
S_0(0) = &1\nonumber
\end{align}
and a defaultable asset $S_1$ with price dynamics:
\begin{align}\label{eq:S_1}
d^-S_1 (t) =&\; S_1 (t-) \Big( \mu (t) dt + \sigma(t) d^-W (t) \\
&+ \int_{\mathbb{R}_0} \theta(t,z) \tilde{N}(d^-t,dz) + \kappa(t) dH(t) \Big), \quad S_1 (0)>0.\nonumber
\end{align}
Here $W(t)$, $t \geq 0$, is a standard Wiener process and $N(dt,dz)$, $t\geq 0$, $\mathbb{R}_0:=\mathbb{R}\setminus \{0\}$ is a Poisson random measure, independent of $W$ and with $\mathbb{E} [N(dt,dz)] = \nu(dz)dt$. We denote $\tilde N(dt,dz) := N(dt,dz) - \nu(dz) dt$.
Moreover $H(t)$, $t \geq 0$ is a c\`adl\`ag counting process, with
\begin{equation*}
\E \big[ H(T) \big] < \infty \quad \text{and} \quad \Prob\big( \Delta H(t) >1 \text{ for any $t\in[0,T]$} \big) = 0.
\end{equation*}
We remark that $H$ is not necessarily independent of $N$ and $W$. Being $H$ a process of finite variation the corresponding integral is intended path-wise. On the other side, the $d^-$ indicates forward integration. The forward integral extends the It\^o integral but does not require the integrand to be adapted to a specific filtration, see Section \ref{section:forward_integrals} for details.

The random processes considered live in a complete probability space $(\Omega, \mathcal{A}, \mathbb{P})$. In the sequel the following $\Prob$-augmented filtrations appear
\begin{itemize}
\item
$\mathbb{F}^H := \big\{ \mathcal{F}^H_t \subset \mathcal{A}, \: t \geq 0 \big\}$ where
$\mathcal{F}^H_t = \sigma \big\{ H(s), \: s \leq t\big\}$,
\item
$\mathbb{F} := \big\{ \mathcal{F}_t \subset \mathcal{A}, \: t \geq 0 \big\}$ where
$\mathcal{F}_t = \sigma \big\{ W(s), N((s,t],B), \: s \leq t,\, B\in \mathcal{B}(\mathbb{R}_0) \big\}$,
\item
$\mathbb{G} := \big\{ \mathcal{G}_t \subset \mathcal{A}, \: t \geq 0 \big\}$ where
$\mathcal{G}_t$ is a right continuous filtration that represents the information available to the investor at time $t$.
\end{itemize}
We assume that the coefficients $\rho$, $\mu$, $\sigma$, and $\kappa$ are c\`agl\`ad
stochastic processes and $\theta$ is a c\`agl\`ad random field, in the sense that $\theta(\cdot,z)$ is c\`agl\`ad $\nu$-a.e. ($\mathbb{P}$-a.e.). Here $\rho$, $\mu$, $\sigma$ and $\kappa$ are measurable with respect to $\mathcal{A}\times  \mathcal{B}([0,T])$ while $\kappa$ is $\mathcal{A}\times  \mathcal{B}([0,T]\times \mathbb{R}_0)$-measurable. The choice of the forward integral in \eqref{eq:S_1} allows us to drop the usual requirements of adaptedness of the coefficients to the given information. Naturally, in the case of adaptedness \eqref{eq:S_1} could be expressed in terms of It\^o integration (see Section \ref{section:forward_integrals}).

The Borel measure $\nu(dz)$ on $\mathbb{R}_0$ is $\sigma$-finite and satisfies $\int_{\mathbb{R}_0}  z^2 \nu(dz) < \infty$. For modeling purposes $\kappa$ would be taken to be negative though it is not a necessary condition for the optimization problem.

We denote $\Lambda$ as the $\mathbb{G}$-predictable intensity of $H$, \ie the $\mathbb{G}$-predictable random measure such that
\begin{equation*}
\E\Big[ \int\limits_0^t \kappa(s) dH(s) \Big] = \E\Big[ \int\limits_0^t \kappa(s)\Lambda(ds) \Big],
\end{equation*}
for all $\invFF$-predictable processes $\kappa$. 
In addition, we assume that $H$ and $N$ do not jump at the same time, \ie
\begin{align}
\mathbb{P} \big(\text{There exist } t\in[0,T] \text{ and } U \subset \mathbb{R}_0 & \text{ compact}\text{ such that } \nonumber \\
& \Delta H(t) >0  \text{ and }  N\big( \Delta t,U \big) > 0 \big) = 0.
\label{assumption:no_common_jumps}
\end{align}

We set $\sigma$ forward integrable with respect to $W$, $\theta$ and $\ln(1+\theta)$ forward integrable with respect to $N$ and
\begin{equation}
\mathbb{E} \Big[ \int\limits_0^T  \big| \mu(s) \big| + \big| \sigma(s) \big| ^2 + \int_{\mathbb{R}_0} \big| \theta(s,z)  \big|^2 \nu(dz)\, ds + \int\limits_0^T \big| \kappa(s) \big|\Lambda(ds) \Big] < \infty.
\label{eq:finite_default_integrals}
\end{equation}
To have $S_1$ well defined and non-negative at all times, we assume 
\begin{align}
-1 &< \theta(t,z,\omega)\quad\quad  dt\times \nu(dz)\times d\mathbb{P} \text{ a.e.} \\
-1 &\leq  \kappa(t,\omega)  \quad\quad  dt\times d\mathbb{P} \text{ a.e.}   
\label{eq:kappa_bounded}
\end{align}
%
%
Using an adequate version of the It\^o formula (Theorem \ref{teorem:newItoFormula}), we see that the solution of \eqref{eq:S_1} is
\begin{align}\label{eq:S1_solution} 
S_1(t) =&\; S_1(0)  \prod_{ \substack{ \Delta H(s)>0 \\ s\leq t }} \bigg( 1 +   \kappa(s)  \Delta H(s) \bigg)  \exp \bigg\{ \int\limits_0^{t} \big[ \mu(s) -\frac{1}{2} \sigma^2(s) \big] ds  \\
&+ \int\limits_0^{t} \sigma(s) d^-W(s) - \int\limits_0^{t} \int\limits_{\mathbb{R}_0} \big[ \ln\big(1+\theta(s,z)\big)- \theta(s,z) \big] \nu(dz)\, ds \nonumber \\
&+ \int\limits_0^{t} \int\limits_{\mathbb{R}_0} \ln\big(1+\theta(s,z)\big)\tilde{N}(d^-s,dz) \bigg\} \nonumber
\end{align}
and it is easy to argue that this solution is unique. This can be achieved using similar arguments as in \cite[Theorem 37]{Protter2005} though adapted to forward integration.

\bigskip

The investor's optimization problem is to divide his money between the asset $S_1$ and the bond $S_0$ in order to achieve the maximum expected utility of the portfolio value at the end of the period allowed. The investor bases his decisions on the information available to him represented by the filtration $\invFF$. 
The investor's wealth $\tilde{X}_\pi(t)$, $t \in [0,T]$, is given by:
\begin{equation}
d\tilde{X}(t) = (1-\pi(t)) dS_0(t) + \pi(t) d^-S_1(t)
\label{eq:investor_sde}
\end{equation}
with initial value $\tilde{X}(0) = x_0>0$. The process $\pi(t)$, $t \in [0,T]$, represents the fraction of wealth invested in $S_1$.
Note that $\pi$ is a $\mathbb{G}$-adapted stochastic process.

We aim for generality in how the optimization scheme ends. In particular we are interested in the two different scenarios:
\begin{enumerate}
\item It is no longer possible to invest in the asset $S_1$ after the first jump of $H$. 
In this case, the jump of $H$ signifies default or another catastrophic event. See, e.g. \cite{Scalliet2008,Bouchard2003,DeLong2006}.
\item It is possible to invest in the asset $S_1$ even after several ``default'' events. The jumps of $H$ signify the occurence of these ``default'' events. 
and the dynamics of $S_1$ can possibly change. See, e.g. \cite{ElKaroui2010,pham2010stochastic}. 
\end{enumerate}
To describe both the above scenarios, we assume that the it is longer possible to invest in $S_1$ after a $\invFF$-stopping time $\tau \leq T$. For the period $(\tau,T]$ all the investor's wealth is invested in the bond. The stopping time $\tau$ must satisfy $\tau \leq T$, meaning that the optimization problem terminates in any case when the time horizon is reached, and $\tau \leq \inf\{ t\in [0,T] : S_1(t) = 0\}$, meaning that the optimization problem ends if there is no value in the asset $S_1$. 




By application of the It\^o formula, we can see that the (unique) solution of \eqref{eq:investor_sde}, for a given admissible $\pi$ (see Definition \ref{definition:Allowable_controls1}), is:
\begin{align}
\tilde{X}_{\pi}(t) &= x_0 \exp \Big\{ \int\limits_0^{t} \Big[ \rho(s) + \big(\mu(s)-\rho(s) \big) \pi(s) - \frac{1}{2} \sigma^2(s) \pi^2(s) \Big] ds \nonumber \\
+&\; \int\limits_0^{t} \int\limits_{\mathbb{R}_0} \Big[ \ln\big(1+\pi(s)\theta(s,z)\big)- \pi(s)\theta(s,z) \nu(dz) \Big] ds \nonumber + \int\limits_0^{t} \sigma(s) \pi(s) d^- W(s) \nonumber  \\
+&\; \int\limits_0^{t} \int\limits_{\mathbb{R}_0} \ln\big(1+\pi(s)\theta(s,z)\big)\tilde{N}(d^-s,dz)  + \int\limits_0^{t} \ln\big(1+ \kappa(s) \pi(s) \big) dH(s)  \Big\}.
\label{eq:investor_sde_solution}
\end{align}
and set $X_\pi(T) := \tilde{X}_{\pi}(\tau) e^{\int_\tau^T \rho(s) ds}$.

In summary we study the optimal portfolio problem
\begin{equation}\label{OP}
\sup_{\pi\in\mathcal{A}_{{\mathbb{G}}}} \mathbb{E} \big[ U \big( \tilde{X}_{\pi}(\tau) e^{\int_\tau^T \rho(s) ds} \big) \big] = \sup_{\pi\in\mathcal{A}_{{\mathbb{G}}}} \mathbb{E} \big[ U \big( X_{\pi}(T) \big) \big],
\end{equation}
of an investor having $\mathbb{G}$ as information flow at disposal and $U$ as utility function. Here $\mathcal{A}_{{\mathbb{G}}}$ represents the set of admissible portfolios (see Definition \ref{definition:Allowable_controls1}).
%

\bigskip
The optimization scheme itself Theorem \ref{theorem:localMax} and the related Theorem \ref{theorem:semi-martingale} are an extension of the results in \cite{Biagini2005, DiNunno2005} to include a form of default risk.
We refer to \cite{Russo1993, Russo1995} for the treatment of the forward integral with respect to the Wiener process, and to \cite{DiNunno2005} for the case of integration with respect to the compensated Poisson random measure. The forward integral is an extension of the It\^o integral, but does not require the adaptedness of the integrands to the integral filtration.
Applications of this type of integration to optimization problems and the justification of the use of these integrals from the modeling point of view have been studied. See, e.g.  \cite{Biagini2005,DiNunno2006, Okur2009, Kohatsu-Higa2006}. We also refer to \cite{DiNunno2009} for a unified presentation of the topics.

\medskip

Related to our optimization problem is the optimization of investments under uncertain time-horizons, as done in \cite{Scalliet2008,DeLong2006}. 
In \cite{Scalliet2008}, optimization ends at a stopping time $\tau$ related to the noise in stock price. 
In \cite{DeLong2006} both optimal consumption and investment are treated. Typically the problems are solved using some variants of Hamilton-Jacobi-Bellman (HJB) equations.
Our approach differs from these works for several reasons. 
First we focus on different streams of information for the investor, second we consider that the loss in the case of default depends on the position in the risky asset. Moreover, our approach is different in framework and we do not use HJB type solutions. 
In \cite{Lim2009} we find a study of a problem similar to ours. The approach is however entirely different as in this case backward stochastic differential equations are involved.
Moreover we allow for a more general information structure and we consider a 
L\'evy type of noise in the price dynamics.

Our work has some similarities to \cite{Bauerle2007}, where an optimization problem is considered when the stock dynamics include a jump component with an unknown intensity modeled by a continous time Markov chain. But the filtering techniques therein may be less suited to default modeling since default is a jump happening only once. The methodology presented there relies on HJB equations and differs from ours.

Bielecki and coauthors consider various forms of optimal investments in, e.g. \cite{Bielecki2004}, \cite{Bielecki2007} and \cite{Bielecki2006a}, looking at optimality and hedging when there is a number of instruments, some of which are subject to default. 
However, their main focus is on the use of defaultable instruments for hedging purposes and the evaluation on whether to invest in defaultable bonds. In the same line is the study in \cite{Hou2002}.

As announced, in this paper we adopt the framework of anticipating stochastic calculus, specifically forward integration to tackle the optimization problem \eqref{OP}. Moreover, we consider the problem for various choices of investor's information flow $\mathbb{G}$.
To the best of our knowledge it is the first time that the framework of forward integration is applied in optimization problems in presence of default.

\medskip
In this paper we provide a characterization for the existence of locally optimal controls in a great generality both in the choice of utility function and in the amount of information available. Considerations on the meaning of locality and some examples are also provided. These topics are presented in Section 3. The key results of forward integration is summarized in Section 2. In Section 4 we reinterpret the results of section 3 in the context of semimartingale-integration.


\section{Mathematical framework: Forward Integrals}
\label{section:forward_integrals}

Forward integrals were introduced by Russo and Valois in the articles \cite{Russo1993} and \cite{Russo1995} for continuous processes and in \cite{DiNunno2005} for pure jump L{\'e}vy process, see also \cite{DiNunno2009} for a systematic presentation. 

The forward integral is a type of stochastic anticipating integration that does not require assumptions of adaptedness or predictability to some filtration related to the integrator. Moreover, it is also an extension of the It\^o integral in the sense that when the appropiate predictability is in place the two integrals coincide. This makes the forward integral especially suited for studying portfolio optimization problems under insider or partial information, where different filtrations are considered. See for, e.g. \cite{Biagini2005, DiNunno2005} and \cite{DiNunno2009}.

We follow the idea of \cite{Kohatsu-Higa2006} and consider the forward integral with respect to the Wiener processes as a limit in $L^1(\mathbb{P})$. This would also imply forward integrability in the sense of Russo and Valois, \cite{Russo1993,Russo1995,Russo2000}, who consider the same limit in probability.


\begin{definition}
We say that the stochastic process $\sigma = \sigma(t,\omega),\,t\in[0,T],\,\omega\in\Omega$, is forward integrable over the interval $[0,T]$ with respect to W if there exists a process $I = I(\sigma,t), t\in[0,T]$, such that 
\begin{equation*}
\mathbb{E} \Big[ \sup_{t\in[0,T]} \Big| \int\limits_0^t \sigma(s) \frac{W(s+\epsilon) - W(s)}{\epsilon} ds - I(\sigma,t) \Big| \Big] \longrightarrow 0, \quad \text{as } \epsilon \to 0^+, 
\end{equation*}
In this case we write 
\begin{equation*}
I(\sigma,t) = \int\limits_0^t \sigma(s) d^- W(s), \; t\in[0,T], 
\end{equation*}
and call $I(\sigma,t)$ the forward integral of $\sigma$ with respect to W on $[0,t]$.
\label{Giulia_8_3}
\end{definition}

Lemma \ref{Giulia_8_9} shows that the forward integral is an extension of the It\^o integral.
%
\begin{lemma}
Let $\mathbb{G} = \{ \mathcal{G}_t, t\in [0,T]$ \} be a given filtration. Suppose that
\begin{enumerate}
\item W is a semimartingale with respect to the filtration $\mathbb{G}$,
\item $\sigma$ is $\mathbb{G}$-predictable and the It\^o integral $\int\limits_0^T \sigma(t) dW(t)$ exists (in $L^1(\mathbb{P})$),
\end{enumerate}
then $\sigma$ is forward integrable and 
\begin{equation*}
\int\limits_0^T  \sigma(t) d^-W(t) = \int\limits_0^T \sigma(t) dW(t).
\end{equation*}
\label{Giulia_8_9}
\end{lemma}
\noindent For proof we refer to, e.g. \cite[Lemma 8.9]{DiNunno2009}.

Elementary processes are forward integrable, and have a natural interpretation as Riemann-like sums. Suppose the stochastic process $\sigma = \sigma(t,\omega)$, $t\in [0,T]$,$ \omega\in\Omega$, is elementary, meaning that it has the form 
\begin{equation}
\sigma(t,\omega) = \sum_{i=0}^{N-1} \sigma_{i}(\omega) \mathbf{1}_{(t_i,t_{i+1}]}(t),
\label{eq:elementary_function}
\end{equation}
where the $\sigma_{i}$ are bounded random variables and $0 = t_0 < t_1 \dots < t_{N} = T$. Then $\sigma$ is forward integrable, see  \cite[Remark 1]{Russo2007}, and  
\begin{equation}
\int\limits_0^t \sigma(s) d^-W(s) = \sum_{i=0}^{N-1} \sigma_{i}  \Big( W(t_{i+1} \wedge t) - W(t_i \wedge t) \Big), \quad t\in[0,T].
\label{eq:elementaryForward} 
\end{equation}
However, it is not obvious that one can approximate a general forward integrable function by elementary functions and in this way obtain also an approximation to the integral.


\begin{example}\label{example:sums_to_infinity}
Let $A=\big\{f \in L^{\infty}( [0,1]\times\Omega ) : f \text{ is c{\`a}gl{\`a}d}, \, |f(t,\omega)| \leq 1 \text{ for all } (t,\omega)\in [0,1]\times \Omega\}$. Then any $f\in A$ is forward integrable. 
But 
\begin{equation}
\sup_{f\in A} \mathbb{E} \Big[ \int_0^1 f(s) d^- W(s) \Big] = \infty.
\label{eq:infty}
\end{equation}
So even though $f$ is bounded, the forward integral with respect to $d^- W$ can have arbitrarily large expectations. This would not happen with It{\^o} integrals as it is a result of $W$ having infinite total variation and using anticipating information. 

To prove \eqref{eq:infty}, let $f_n$ be elementary functions of the form
\begin{equation*}
f_n = \sum_{j=0}^{n-1} \text{sign}\Big( W\big(\frac{j+1}{n}\big) - W\big(\frac{j}{n}\big) \Big) \mathbf{1}_{\{ t \in (\frac{j+1}{n}, \frac{j+1}{n}]\}},
\end{equation*}
where 
\begin{equation*}
\text{sign}(x) = \left\{
\begin{array}{rl}
-1 & \text{if } x < 0\\
1 & \text{if } x \geq 0
\end{array} \right. 
\end{equation*}
Then $f_n \in A$ and 
\begin{align*}
\int\limits_0^1 f_n(s) d^- W(s) &= \sum_{j=0}^{n-1} \text{sign}\Big( W\big(\frac{j+1}{n}\big) - W\big(\frac{j}{n}\big) \Big) \Big(W\big(\frac{j+1}{n}\big) - W\big(\frac{j}{n}\big) \Big) \\
&= \sum_{j=0}^{n-1} \Big| W\big(\frac{j+1}{n}\big) - W\big(\frac{j}{n}\big) \Big|.
\end{align*}
We have $\mathbb{E}[ |W(t)-W(s)| ] = \sqrt{(t-s) \cdot 2/\pi}$, (see for instance \cite{Leone1961}), so
\begin{align*}
\mathbb{E} \Big[ \int\limits_0^1 f_n(s) d^- W(s)\Big] &= \mathbb{E} \Big[  \sum_{j=0}^{n-1} \Big| W\big(\frac{j+1}{n}\big) - W\big(\frac{j}{n}\big) \Big| \Big] \\
&= \sqrt{n \cdot 2/\pi}  \longrightarrow \infty \quad \text{as} \; n\to \infty.
\end{align*}
Additionally we can remark that letting $g_n = n^{-1/4} f_n$, we would get that $g_n \to 0$ pointwise and is bounded by the forward integrable function $1$, thus proving that the dominated convergence theorem does not hold for forward integrals with respect to Brownian motions. 
\end{example}

Characterizing when integrals are finite or limits do not explode is non-trivial, and from the remark above we see that the boundedness of the integrand is not enough. Thus we have to be careful, even though a sequence of forward integrable functions converge in some suitable space, the corresponding forward integrals over these functions may not converge at all. See also the discussion in \cite[Section 1.8]{Protter2005}.

\begin{definition}
The forward integral 
\begin{equation*}
J(\theta,t) := \int\limits_0^t \int\limits_{\mathbb{R}_0} \theta(s,z) \tilde{N}(d^-s,dz), \quad t\in[0,T]
\end{equation*}
with respect to the Poisson random measure $\tilde{N}$ of a c\`agl\`ad random field $\theta(t,z,\omega)$, $t\in [0,T]$, $z\in\mathbb{R}_0$, $\omega\in\Omega$, is defined as
\begin{equation*}
J(\theta,t) = \lim_{m\to\infty} \int\limits_0^t \int\limits_{\mathbb{R}_0} \theta(s,z) \mathbf{1}_{U_m} \tilde{N}(ds,dz)
\end{equation*}
if the limit exists in $L^2(\mathbb{P})$. Here, $U_m, \, m=1,2,\dots,$ is an increasing sequence of compact sets $U_m \subset \mathbb{R}_0$ with $\nu(U_m) < \infty$ such that $lim_{m\to\infty} U_m = \mathbb{R}_0$.
\label{Giulia_15_1}
\end{definition}
Also in this case the forward integral is an extension of the It\^o integral \cite[Remark 15.2]{DiNunno2009}:
%
\begin{remark}
Let $\mathbb{G} = \{\mathcal{G}_t, \, t\in[0,T]\}$ be a given filtration such that
\begin{enumerate}
\item The process $\eta(t) = \int_0^t\int_{\mathbb{R}_0} z \tilde{N}(ds,dz), \, t\in [0,T]$, is a semimartingale with respect to $\mathbb{G}$.
\item The random field $\theta = \theta(t,z),\, t\in [0,T],\, z\in \mathbb{R}_0$, is $\mathbb{G}$-predictable.
\item The integral $\int_0^t\int_{\mathbb{R}_0} \theta(t,z) \tilde{N}(ds,dz)$ exists as a classical It\^o integral.
\end{enumerate}
Then $\theta$ is forward integrable and we have
\begin{equation*}
\int\limits_0^T \int\limits_{\mathbb{R}_0} \theta(s,z) \tilde{N}(d^-t,dz) = \int\limits_0^T \int\limits_{\mathbb{R}_0}  \theta(s,z) \tilde{N}(dt,dz).
\end{equation*}
\label{Giulia_15_2}
\end{remark}

\subsection{The It\^o formula for forward integrals}
\label{subsection_itoformula}

An It\^o formula for forward type integrals when the integrator is continuous was developed in \cite{Russo1995, Russo2000}. An It\^o formula for forward integrals with Poisson random measures is found in \cite{DiNunno2005}, both the results are also summarized in \cite{DiNunno2009}. In this paper we need a more general version that include processes of finite variation to guarantee the existence of solutions of \eqref{eq:investor_sde} and \eqref{eq:S_1}. The proof can be seen as a continuation of the one presented in \cite[Theorem 8.12]{DiNunno2009}, thus we only sketch the additional part. 

\begin{theorem} ~\newline
Let 
\begin{equation*}
d^- X(t) =  \mu(t) dt + \sigma(t) d^- W(t) + \int\limits_{\mathbb{R}_0} \theta(t,z) \tilde{N}(d^- t,dz) + d\zeta(t),
\end{equation*}
where 
\begin{itemize}
\item $\mu$ is a stochastic process satisfying $\int\limits_0^T \big| \mu(s) \big| ds < \infty \quad \mathbb{P}\text{-a.s.} $
\item $\sigma$ is forward integrable with respect to $W$.
\item $\theta$ and $|\theta|$ are forward integrable with respect to $\tilde{N}$ and $\theta$ satisfies \\ $\int\limits_0^T \int_{\mathbb{R}_0} \big| \theta(s,z)\big|^2 \nu(dz) \,ds < \infty \quad \mathbb{P}\text{-a.s.}$
\item $\zeta$ is a c{\`a}dl{\`a}g pure jump process of finite variation, with $\zeta(0) = 0$ and
\begin{align}
\mathbb{P} \big(\text{There exist } t\in[0,T] \text{ and } U \subset \mathbb{R}_0 & \text{ compact}\text{ such that } \nonumber \\
& \Delta \zeta(t) >0  \text{ and }  N\big( \Delta t,U \big) > 0 \big) = 0.
\label{eq:no_similar_jumps}
\end{align}
for all $U \subset \mathbb{R}_0$ compact. 
Here $N\big( \Delta t,U \big) := N\big( (0,t],U\big)-N\big( (0,t),U\big)$ and $\Delta \zeta(t) := \zeta(t)-\zeta(t-)$.
\end{itemize}

Assume $f\in C^2\big(\mathbb{R}\big)$ and let $Y(t) = f\big(X(t)\big)$. Then 
\begin{align*}
Y(t) &= Y(0) + \int\limits_0^t \Big[ f'\big(X(s-)\big) \mu(s) + \frac{1}{2} f''\big(X(s-)\big) \sigma^2(s) \Big] ds \\
&+ \int\limits_0^t \int\limits_{\mathbb{R}_0} \Big[ f\big(X(s-)+\theta(t,z)\big) - f(X(s-)\big) - f' \big(X(s-)\big)  \theta(s,z) \Big] \nu(dz)\, ds \\
&+ \int\limits_0^t  f'\big(X(s-)\big) \sigma(s) d^-W(s) + \int\limits_0^t \int\limits_{\mathbb{R}_0} \Big[ f\big(X(s-)+\theta(s,z)\big) - f(X(s-)\big)\Big] \tilde{N}(d^-s,dz) \\
&+ \sum_{\substack{0<s<t \\ \Delta \zeta(s) \neq 0}} \Big[ f\big(X(s-) + \Delta \zeta(s) \big) - f\big(X(s-) \big) \Big].
\end{align*}
\label{teorem:newItoFormula}
\end{theorem}

\begin{remark}
Condition \eqref{eq:no_similar_jumps} is for instance fulfilled if $N$ and $\zeta$ are independent.
\end{remark}

\begin{proof}
Let 
\begin{equation*}
X_m(t) = x+ \int\limits_0^t \mu(s) ds + \int\limits_0^t \sigma(s) d^- W(s) + \int\limits_0^t \int\limits_{\mathbb{R}_0} \mathbf{1}_{U_m}(z) \theta(s,z) \tilde{N}(d^- s,dz) + \zeta(t),
\end{equation*}
where $\mathbf{1}_{U_m}$ is as in Definition \ref{Giulia_15_1}. We denote $\alpha_i$, $i=1,2 \dots$ the times of the jumps of $X_m$. By condition \eqref{eq:no_similar_jumps} we can uniquely ($\mathbb{P}$-a.s.) divide the sequence $\alpha_i$ by the jumps of either $\zeta$ or $\mathbf{1}_{U_m}(z) N(dt,dz)$ as $\alpha_i^\zeta$ and $\alpha_i^N$. We formally set $\alpha_0 = \alpha_0^\zeta = \alpha_0^N = 0$.

Then
\begin{align*}
f\big(X_m(t)\big) - f\big(X_m(0)\big) =&\; \sum_i \Big[ f\big(X_m(\alpha_i\wedge t)\big) - f\big(X_m(\alpha_i\wedge t-)\big)\Big]  \\
&+ \sum_i \Big[ f\big(X_m(\alpha_i\wedge t-)\big) - f\big(X_m(\alpha_{i-1}\wedge t)\big)\Big] \\
=&\; \sum_{\alpha_i^\zeta\leq t} \Big[ f\big(X_m(\alpha_i^\zeta )\big) - f\big(X_m(\alpha_i^\zeta-)\big)\Big] \\
&+ \sum_{\alpha_i^N \leq t} \Big[ f\big(X_m(\alpha_i^N)\big) - f\big(X_m(\alpha_i^N-)\big)\Big] \\
&+ \sum_{\alpha_i\leq t} \Big[ f\big(X_m(\alpha_i-)\big) - f\big(X_m(\alpha_{i-1}-)\big)\Big] \\
=&\; J_1(t) + J_2(t) + J_3(t),
\end{align*}
with
\begin{equation*}
J_1(t) =  \sum_{\substack{0<s<t \\
\Delta \zeta(s) \neq 0}} \Big[ f\big(X_m(s-) + \Delta \zeta(s) \big) - f\big(X_m(s-) \big) \Big]
\end{equation*}
and
%
\begin{align}
J_2(t) =& \sum_i \big[ f\big(X_m(\alpha_i^N\big) - f\big(X_m(\alpha_i^{N} - \big) \big] \mathbf{1}_{\{\alpha_i^N\leq t\}} \label{eq:itosomething} \\
=&\; \int\limits_0^t \int\limits_{\mathbb{R}_0} \Big[ f\big(X_m(s-)+\theta(s,z)\big) - f(X_m(s-)\big)\Big] N(ds,dz) \nonumber \\
=&\; \int\limits_0^t \int\limits_{\mathbb{R}_0} \Big[ f\big(X_m(s-)+\theta(s,z)\big) - f(X_m(s-)\big)\Big] \tilde{N}(ds,dz) \nonumber \\
&+ \int\limits_0^t \int\limits_{\mathbb{R}_0} \Big[ f\big(X_m(s-)+\theta(s,z)\big) - f(X_m(s-)\big)\Big] \nu(dz)\, ds . \nonumber
\end{align}
For the elements of the sum in $J_3(t)$ we use \cite[Theorem 8.12]{DiNunno2009}:
\begin{align*}
J_3(t) =&\; \sum_i \bigg[ \int\limits_{\alpha_{i-1}\wedge t}^{\alpha_i\wedge t} \Big[  f'\big(X_m(s-)\big) \mu(s) ds -  \int\limits_{\mathbb{R}_0} f' \big(X_m(s-)\big) \mathbf{1}_{U_m} \theta(s,z) \nu(dz)\,\Big]  ds \\
&+ \int\limits_{\alpha_{i-1}\wedge t}^{\alpha_i\wedge t}  f'\big(X_m(s-)\big) \sigma(s) d^-W(s) + \int\limits_{\alpha_{i-1}\wedge t}^{\alpha_i\wedge t}  f''\big(X_m(s-)\big) \sigma^2(s) ds \bigg] \\
=&\; \int\limits_0^t \Big[ f'\big(X_m(s-)\big) \mu(s) + \frac{1}{2}f''\big(X_m(s-)\big) \sigma^2(s) - \int\limits_{\mathbb{R}_0} f' \big(X(s-)\big) \mathbf{1}_{U_m} \theta(s,z) \nu(dz)\, \Big] ds \\
&+ \int\limits_0^t  f'\big(X_m(s^-)\big) \sigma(s) d^-W(s).
\end{align*}
Adding $J_1$, $J_2$ and $J_3$ together and letting $m\to\infty$ the result follows.

\end{proof}

\section[Local maximums]{Optimization problem: local maxima}
\label{section:localMaximums}

Now we are ready to tackle directly our stated optimization problem \eqref{OP}. First we give a description of the set of the investor's admissible portfolios.


\begin{definition} The set $\mathcal{A}_{\invFF}$ of \emph{admissible portfolios} consists of stochastic processes $\pi = \pi(t,\omega),\, t\in [0,T],\, \omega \in \Omega$, such that
\begin{enumerate}
\renewcommand{\theenumi}{\roman{enumi})}
\renewcommand{\labelenumi}{\theenumi}
\item $\pi$ is c\`agl\`ad and $\invFF$-adapted,
\label{list:pi_def1}
\item for every $\pi \in \mathcal{A}_{\invFF}$, there exists $\epsilon_\pi \ > 0$ such that for all t, 
\begin{equation}
\pi(t)\kappa(t) > -1 + \epsilon_\pi
\label{eq:pi_limitation}
\end{equation}
and
\begin{equation}
\pi(t)\theta(t,z) > -1+\epsilon_\pi,
\label{eq:pi_limitation2}
\end{equation}
\label{list:pi_def2}
\item 
\begin{equation*}
\mathbb{E} \Big[ \int\limits_0^T \big| \big(\mu(s)-\rho(s)\big) \big|\big| \pi(s) \big| + \sigma^2(s)\pi^2(s) ds \Big] <  \infty
\end{equation*}
and
\begin{equation*}
\mathbb{E} \Big[ \int_0^T \int_{\mathbb{R}_0} \big| \theta(s,z) \pi(s) \big|^2 \nu(dz)\, ds \Big] < \infty,
\end{equation*}
\label{list:pi_def3}
\item $\pi \sigma$ is c\`agl\`ad and forward integrable with respect to W,
\label{list:pi_def4}
\item $\pi \theta$, $\ln\big(1+\pi\theta\big)$ and $\frac{\pi \theta}{1+\pi  \theta}$ are c\`agl\`ad and forward integrable with respect to $\tilde{N}$.
\label{list:pi_def5}
\end{enumerate}
The subset $\mathcal{A}_{\invFF}^e$ of $\mathcal{A}_{\invFF}$ consists of all admissible portfolios that are representable as elementary integrands- see \eqref{eq:elementary_function}.
\label{definition:Allowable_controls1}
\end{definition}

In particular we note that condition \ref{list:pi_def1} ensures that the portfolio choices correspond to the investors knowledge and that condition \ref{list:pi_def2} ensures that the investor never reaches zero wealth from the jumps of $H$ or $\tilde{N}$. 
In addition \ref{list:pi_def2}  means that fractions of the form $\frac{1}{1+\kappa\pi}$ are bounded, which is implicitly used in some forthcoming equations. 

Note that if 
\begin{equation*}
\pi(s,\omega) = \alpha(\omega) \mathbf{1}_{(t,t+h]}(s),
\end{equation*}
where $\alpha$ is a bounded $\invF_t$-measurable random variable, then $\pi \in \mathcal{A}_\mathbb{\invFF}^e \subset \mathcal{A}_\mathbb{\invFF}$ as long as \eqref{eq:pi_limitation} and \eqref{eq:pi_limitation2} are satisfied.

\medskip
\noindent
As announced we are interested in the problem 
\begin{equation}
\sup_{\pi\in\mathcal{A}_{\invFF}} \mathbb{E} \big[ U \big( X_{\pi}(T) \big) \big]. 
\label{eq:Levy_opt_prob} 
\end{equation}
(We recall that $X_\pi(T) = \tilde{X}_{\pi}(\tau) e^{\int_\tau^T \rho(s) ds}$ is the value of the investor's wealth at $T$ and the definition of $\tilde{X}_{\pi}$ is in \eqref {eq:investor_sde}).
We will search for solutions to \eqref{eq:Levy_opt_prob} that are optimal in the sense that they cannot be improved by small perturbations. 

\begin{definition}We say that the stochastic process $\pi$ is a \emph{local maximum} for the problem \eqref{eq:Levy_opt_prob} if 
\begin{equation}
\mathbb{E} \big[U\big( X_{\pi +y \beta} (T) \big) \big] \leq \mathbb{E} \big[U\big( X_{\pi} (T) \big) \big] 
\label{eq:loxal_max}
\end{equation}
for all bounded $\beta \in \mathcal{A}_{\invFF}$ and $|y| < \delta_{\pi,\beta}$ for some $\delta_{\pi,\beta}>0$ that may depend on $\beta$. 
We say that $\pi$ is a \emph{weak local maximum} for \eqref{eq:Levy_opt_prob} if \eqref{eq:loxal_max} is true for all $\beta\in\mathcal{A}_{\invFF}^e$.
\label{Definition:localMax}
\end{definition}

From the terminology point of view, when we say that a \emph{property holds under $(\mathbb{Q}, \invFF)$}, we mean that the property holds under the measure $\mathbb{Q}$ with respect to the filtration $\invFF$. Moreover, we say that a stochastic process $Y$ has the \emph{martingale property} under  $(\mathbb{Q}, \invFF)$ if 
\begin{equation*}
\mathbb{E}_\mathbb{Q} \big[ Y(t+h) - Y(t) \big| \invF_t \big] = 0
\end{equation*}
for all $0<t<t+h<\infty$. We stress that $Y$ does not need to be a $(\mathbb{Q}, \invFF)$-martingale despite having the martingale property under $(\mathbb{Q}, \invFF)$. In fact no statement is given about $Y$ being adapted to $\invFF$.

Following the techniques in \cite{Biagini2005,DiNunno2006}, we consider pertubations of stochastic controls to find necessary and sometimes sufficient criteria to characterize local maximums. We will need the following assumption for a differentiable utility function $U$.

\smallskip
\noindent \textbf{Assumption \aui.} We say that assumption \aui holds for $\pi \in \mathcal{A}_\mathbb{\invFF}$ if
\renewcommand{\theenumi}{\roman{enumi})}
\renewcommand{\labelenumi}{\theenumi}
\begin{enumerate}
\item $\mathbb{E}[U(X_\pi(T))] < \infty$,
\label{list:pi_def7}
\item $0 < \mathbb{E}[U'(X_{\pi}(T) ) X_{\pi}(T)] < \infty$, with $U'(x) = \frac{dU}{dx}(x)$,
\label{list:pi_def6}
\item For all $\beta \in \mathcal{A}_{\invFF}$ with $\beta$ bounded, there exists $\delta_{\pi,\beta}>0$ that may depend on $\beta$ such that the family
 \begin{equation}
  \big\{U'\big( X_{\pi+y\beta}(T) \big) X_{\pi+y\beta}(T) |\uniproc(y,\beta,\pi)|\big \}_{y\in(-\delta_{\pi,\beta},\delta_{\pi,\beta})}
  \label{eq:uniformIntegrable} 
 \end{equation}
 is uniformly integrable, where
\begin{align}
 \uniproc( & y, \beta,\pi) := \int\limits_0^{\tau} \beta(s) \big[ \mu(s) - \rho(s) - \big(\pi(s)+y\beta(s)\big) \sigma^2(s) \big] ds \nonumber \\
&+ \int\limits_0^{\tau} \int\limits_{\mathbb{R}_0} \Big[ \frac{\beta(s)\theta(s,z)} {1+\big(\pi(s)+y\beta(s)\big)\theta(s,z) } - \beta(s)\theta(s,z) \Big] \nu(dz)\,ds \nonumber \\
&+ \int\limits_0^{\tau} \beta(s) \sigma(s) d^- W (s) + \int\limits_0^{\tau} \int\limits_{\mathbb{R}_0} \frac{\beta(s)\theta(s,z)}{1+\big(\pi(s)+y\beta(s)\big)\theta(s,z) } \tilde{N}(d^-z,ds)  \nonumber \\
&+ \int\limits^{\tau }_0 \frac{\beta(s)\kappa(s)}{1+\kappa(s)\big(\pi(s)+y\beta(s)\big)} dH(s) .
\label{eq:Nybpi}
\end{align} 
\label{list:pi_def8}
\end{enumerate}

\medskip
Assumption \aui depends strongly on the utility function $U$. 
Condition \ref{list:pi_def7} is related to the optimization problem \eqref{eq:Levy_opt_prob} and \ref{list:pi_def6} is used in the definition of \eqref{eq:Fpi}.
Condition \ref{list:pi_def8}, uniform integrability,  is the minimal condition for taking limits under the integral sign.
It is unfortunate in that it stems from mathematical rather than modeling necessities, but we cannot do without it. 
There is a good discussion in when uniform integrability conditions like Assumption \aui is fulfilled in \cite[section 16.5]{DiNunno2009}. The conclusions from \cite[section 16.5]{DiNunno2009} can be transferred to our model. 

\begin{theorem}
Suppose the utility function $U$ is increasing and differentiable, $\pi \in \mathcal{A}_{\invFF}$ and \aui holds.
\begin{enumerate}
\renewcommand{\theenumi}{\roman{enumi})}
\renewcommand{\labelenumi}{\theenumi}
\item If $\pi$ is a local maximum for \eqref{eq:Levy_opt_prob}, then the process $M_{\pi}(t)$, $t\in [0,T]$, has the martingale property under  $(\mathbb{Q}_{\pi}, \invFF )$. Where $M_\pi$ is defined as
\begin{align}
M_{\pi}(t) :=&\;  \int\limits_0^{t\wedge\tau} \Big[ \mu(s) -\rho(s) - \pi(s) \sigma^2(s) -\int\limits_{\mathbb{R}_0} \frac{ \pi(s)\theta^2(s,z)}{1+\pi(s)\theta(s,z)} \nu(dz) \Big]  ds \nonumber \\
&+ \int\limits_0^{t \wedge\tau} \sigma(s) d^- W (s) + \int\limits_0^{t\wedge\tau}\int\limits_{\mathbb{R}_0} \frac{\theta(s,z)}{1+\pi(s)\theta(s,z)}\tilde{N}(d^-s,dz) \nonumber \\
&+ \int\limits^{t\wedge\tau}_0 \frac{\kappa(s)}{1+\kappa(s)\pi(s)} dH(s),
\label{eq:MpiDefinition}
\end{align}
and the measure $\mathbb{Q}_{\pi}$ is defined by $d\mathbb{Q}_{\pi} = F_{\pi}(T) d\mathbb{P}$, with 
\begin{equation}
F_{\pi}(T) = \frac{U' \big( X_{\pi}(T) \big) X_{\pi}(T) }{\mathbb{E} \big[ U' \big( X_{\pi} \big) X_{\pi}(T) \big] }.
\label{eq:Fpi}
\end{equation}
\label{list:localMax1}
\item Suppose the mappings
\begin{equation}
y \to \mathbb{E} \big[ U\big (X_{\pi+y\beta}(T) \big) \big],\quad y\in(-\delta_{\pi,\beta},\delta_{\pi,\beta}), \;\; (\delta_{\pi,\beta} > 0 )
\label{eq:concaveCondition}
\end{equation}
are concave for all controls $\beta \in \mathcal{A}_{\invFF}^e$ and $|y| < \delta$. 
If $M_\pi$ has the martingale property under $(\mathbb{Q}_{\pi}, \invFF )$ then $\pi$ is a weak local maximum for \eqref{eq:Levy_opt_prob}
\label{list:localMax2}
\item Suppose $M_\pi$ is $\invFF$-adapted and the conditions in \ref{list:localMax2} are satisfied. If $M_\pi$ is a $(\mathbb{Q}_{\pi}, \invFF )$-martingale then $\pi$ is a local maximum for \eqref{eq:Levy_opt_prob}.
\label{list:localMax3}
\end{enumerate}
\label{theorem:localMax}
\end{theorem}

\begin{proof}
\textbf{Part \ref{list:localMax1}} If $\pi$ is a local maximum, then for all bounded $\beta$ we have
\begin{equation}
0 = \frac{d}{dy} \mathbb{E} \big[ U \big( X_{\pi+y \beta}(T) \big) \big]_{|y=0} = \mathbb{E} \big[ U'\big( X_{\pi+y\beta}(T) \big) \frac{d}{dy} X_{\pi+y\beta}(T) \big]_{|y=0}.
\label{eq:interchanged}
\end{equation}
Here assumption \aui is used, see for instance \cite[Appendix A]{Dudley1999}. With some calculations we obtain
\begin{align}
0 =&\; \mathbb{E} \bigg[ U'\big(  X_{\pi}(T)\big) X_{\pi}(T) \Big\{ \int\limits_0^{\tau} \beta(s) \big[\mu(s) - \rho(s) - \pi(s) \sigma^2(s) \big] ds  \nonumber \\
&+\int\limits_0^{\tau} \int\limits_{\mathbb{R}_0} \beta(s) \frac{- \pi(s) \theta^2(s,z)} {1+\pi(s)\theta(s,z) }  \nu(dz)\,ds + \int\limits_0^{\tau} \beta(s) \sigma(s)  d^- W (s) \nonumber \\
&+ \int\limits_0^{\tau} \int\limits_{\mathbb{R}_0} \frac{\beta(s)\theta(s,z)}{1+\pi(s)\theta(s,z)} \tilde{N}(d^-s,dz)  +\int\limits^{\tau}_0 \frac{\beta(s)\kappa(s)}{1+\pi(s)\kappa(s)} dH(s) \Big\} \bigg] \nonumber \\
=&\; \mathbb{E} \bigg[U'\big(X_\pi(T)\big) X_\pi(T) \uniproc(0,\beta,\pi) \bigg].
\label{eq:Levy_varPr1}
\end{align}
We now let $\beta(s) = \alpha \mathbf{1}_{(t, t+h]}(s)$, where $\alpha$ is a $\invF_t$-measurable bounded random variable. We can put $\alpha$ outside the forward integrals, see for instance \cite[Lemma 8.7]{DiNunno2009} and \cite[Remark 15.3]{DiNunno2009} to get
\begin{align}
\mathbb{E} & \bigg[ U'(X_{\pi}(T)) X_{\pi}(T) \Big\{ \int\limits_{t\wedge\tau} ^{(t+h)\wedge\tau} \Big[\mu(s) - \rho(s) - \pi(s) \sigma^2(s) \nonumber \\
&- \int\limits_{\mathbb{R}_0} \frac{\pi(s) \theta^2(s,z)} {1+\pi(s)\theta(s,z) } \nu(dz) \Big] ds + \int\limits_{t\wedge\tau}^{(t+h)\wedge\tau} \sigma(s)  d^- W (s) \nonumber \\
&+ \int\limits_{t\wedge\tau}^{(t+h)\wedge\tau} \int\limits_{\mathbb{R}_0} \frac{\theta(s,z)}{1+\pi(s)\theta(s,z)} \tilde{N}(d^-z,ds) + \int\limits^{(t+h)\wedge\tau}_{t\wedge\tau}  \frac{\kappa(s)}{1+\kappa(s)\pi(s)} dH(s) \Big\}   \alpha \bigg] = 0.
\label{eq:Levy_varPr2}
\end{align}
%
Hence we conclude that
\begin{equation*}
\mathbb{E} \big[ F_{\pi}(T)\big( M_{\pi}(t+h)-M_{\pi}(t) \big) |\invF_t \big] = 0
\end{equation*}
with $F_{\pi}(T)$ and $M_{\pi}$ defined as in \eqref{eq:Fpi} and \eqref{eq:MpiDefinition} respectively. Since $\mathbb{E}[ F_{\pi}(T) ] = 1$, we can define a new probability measure on $(\Omega, \mathcal{A})$ by 
\begin{equation}
d\mathbb{Q}_{\pi} = F_{\pi}(T)d\mathbb{P}.
\label{eq:changeOfMeasure}
\end{equation} 
We thus have that if $\pi$ is a local maximum, $M_{\pi}$ has the martingale property under $( \mathbb{Q}_{\pi},\invFF)$. 

\medskip \noindent \textbf{Part \ref{list:localMax2}}. 
Suppose $M_{\pi}$ has the martingale property under $(\mathbb{Q}_{\pi},\invFF) $. Then, for $0<t<t+h<T$,
\begin{equation*}
\mathbb{E}_{\mathbb{Q}_{\pi}} \big[ M_{\pi}(t+h)-M_{\pi}(t) \big| \invF_t \big] = 0,
\end{equation*}
or, equivalently, that for all bounded $\invF_t$-measurable random variables $\alpha$ we have 
\begin{equation*}
0 = \E_{\mathbb{Q}_{\pi}} \Big[ \alpha \big( M_{\pi}(t+h)-M_{\pi}(t) \big) \Big| \invF_t \Big] = \E_{\mathbb{Q}_{\pi}} \Big[ \int\limits_0^T \alpha \ind_{(t,t+h]}(s) d^- M_\pi (s) \Big| \invF_t \Big].
\end{equation*}
Taking linear combinations we get that 
\begin{equation}
0 = \E \Big[ U'(X_\pi(T) ) X_\pi(T)  \int\limits_0^T \beta(s) d^- M_\pi (s) \Big]
\label{eq:weak_localmax_argument}
\end{equation}
for any $\beta \in \mathcal{A}_{\invFF}^e$. Since the mapping  $y \to \mathbb{E} \big[ U(X_{\pi+y\beta}(T) \big]$ is concave on $|y| < \delta_{\pi,\beta}$ then $\pi$ is a weak local maximum.

\medskip \noindent \textbf{Part \ref{list:localMax3}}. The conditions of \ref{list:localMax2} are satisfied so \eqref{eq:weak_localmax_argument} holds for all $\beta \in \mathcal{A}_{\invFF}^e$. Let $\beta \in \mathcal{A}_{\invFF}$, $\beta$ be bounded, and $\beta_j$, $j=1,\dots$, be a sequence of elementary stochastic processes $\beta_j \in \mathcal{A}_{\invFF}^e$ such $\beta_j$ converges pointwise in $\omega$ and uniformly in $t$ to $\beta$. 

Since $M_\pi$ is adapted and has the martingale property, it is a local martingale and
\begin{equation*}
\int\limits_0^T \beta_j(s) dM_\pi(s) \longrightarrow  \int\limits_0^T \beta(s) dM_\pi(s)  \quad \text{in probability as } j\to \infty.
\end{equation*}
By assumption \aui, the random variable $\int_0^T \beta(s) dM_\pi(s)$ is ${\mathbb{Q}_{\pi}}$-integrable so that 
\begin{equation}
\E_{\mathbb{Q}_{\pi}} \Big[  \int\limits_0^T \beta(s) dM_\pi(s)  \Big] = 0.
\label{eq:general_integral_zero}
\end{equation}
Since the mapping  $y \to \mathbb{E} \big[ U(X_{\pi+y\beta}(T) \big]$ is concave, from the computations in part \ref{list:localMax1} we see that  \eqref{eq:general_integral_zero} can only be zero if $\pi$ is a local maximum.
\end{proof}

With the introduction of the forthcoming assumption $A_{d^2}$ we can detail additional results on the convavity of \eqref{eq:concaveCondition} and the uniqueness of local maximums.

\smallskip
\textbf{Assumption $A_{d^2}$:} The utility function $U$ is twice differentiable, strictly increasing and concave. For $\pi \in \mathcal{A}_\mathbb{\invFF}$, we assume that for all $\beta \in \mathcal{A}_\mathbb{\invFF}$ bounded, there exists a $\delta_{\pi,\beta} > 0$, that may depend on $\beta$, such that the family  
\begin{align*}
\Big\{ &\, U''\big(X_{\pi+y\beta}(T)\big) X^2_{\pi+y\beta}(T) \uniproc^2(y,\beta,\pi) \\
&+ U'\big(X_{\pi+y\beta}(T)\big) X_{\pi+y\beta}(T)\big[ \uniproc(y,\beta,\pi)+\uniproc_y(y,\beta,\pi)\big] \Big\}_{|y|<\delta_{\pi,\beta}}
\end{align*}
%
is uniformly integrable where $\uniproc(y,\beta,\pi)$ is defined in \eqref{eq:Nybpi} and 
\begin{align}
\uniproc_y(y, & \beta,\pi) := \frac{d}{dy} \uniproc(y,\beta,\pi) \nonumber \\
=& -\int\limits_0^{\tau}\beta^2(s)\sigma^2(s) ds  
- \int\limits_0^{\tau}\int\limits_{\mathbb{R}_0} \frac{\beta^2(s)\theta^2(s,z)}{\big[1+\big(\pi(s)+y\beta(s)\big)\theta(s,z)\big]^2} N(d^-s,dz) \nonumber \\
&- \int\limits_0^{\tau} \frac{\beta^2(s)\kappa^2(s)}{\big[1+\big(\pi(s)+y\beta(s)\big)\kappa(s) \big]^2 } dH(s).
\label{eq:NyDerivative}
\end{align}

Since it is reasonable to assume that the coefficients $\sigma$, $\theta$ and $\kappa$ are not zero on the same time intervals, then $\uniproc_y(y,\beta,\pi) < 0$ for $|y| < \delta_{\pi,\beta}$ and $\beta\neq 0$.

\medskip
\noindent

Lemma \ref{lemma:concavity} will give us a sufficient condition for the concavity of \eqref{eq:concaveCondition} in Theorem \ref{theorem:localMax}. 

\begin{lemma}
Suppose \aui and $A_{d^2}$ hold with $|y| < \delta_{\pi,\beta}$, $\beta \in \mathcal{A}_\mathbb{\invFF}$ bounded, and that the utility function $U$ satisfies
\begin{equation}
xU''(x) + U'(x) \leq 0 \quad \text{for all } x>0.
\label{eq:Levy_varPr3}
\end{equation}
Then for $\pi \in \mathcal{A}_\mathbb{\invFF}$ the mappings \eqref{eq:concaveCondition}, $y \to \mathbb{E} \big[ U(X_{\pi+y\beta}(T) \big]$, $y\in(-\delta_{\pi,\beta},\delta_{\pi,\beta})$, $\delta_{\pi,\beta}>0$, are concave for all bounded controls $\beta \in \mathcal{A}_{\invFF}$.
\label{lemma:concavity}
\end{lemma}
\begin{proof}
By assumptions \aui and $A_{d^2}$ the following equations hold true:
\begin{align}
\frac{d^2}{dy^2} & \mathbb{E}\Big[ U\big(X_{\pi+y\beta}(T) \big) \Big] = \nonumber \\
=&\; \frac{d}{dy} \mathbb{E}\Big[ \Big( U'\big(X_{\pi+y\beta}(T)\big) X_{\pi+y\beta}(T)  \uniproc(y,\beta,\pi) \Big) \Big] \nonumber \\
=&\; \mathbb{E} \bigg[ X_{\pi+y\beta}(T) \uniproc^2(y,\beta,\pi) \Big( U''\big(X_{\pi+y\beta}(T)\big) X_{\pi+y\beta}(T) + U'\big(X_{\pi+y\beta}(T)\big) \Big) \nonumber \\ 
&+ U'\big(X_{\pi+y\beta}(T)\big)  X_{\pi+y\beta}(T) \uniproc_y(y,\beta,\pi) \bigg], \quad |y| < \delta_{\pi,\beta}
\label{eq:second_derivative}
\end{align}
Thanks to \eqref{eq:Levy_varPr3} and the observation that $\uniproc_y(y,\beta,\pi)<0$ for all $|y| < \delta_{\pi,\beta}$, both summands are negative and the mapping \eqref{eq:concaveCondition} is locally concave.
\end{proof}

\begin{remark}
Examples of utility functions satisfying \eqref{eq:Levy_varPr3} are the power utility $U(x) = \frac{1}{1-c} x^{1-c}$ when $c>1$, and logarithmic utility $U(x) = log(x)$, while the exponential utility, $U(x) = \frac{-1}{\gamma} e^{-\gamma x}$, does not.
\label{remark:concave}
\end{remark}

\begin{remark}
Condition \eqref{eq:Levy_varPr3} can also be discussed in terms of the \emph{Arrow Pratt measure of relative risk aversion}. This measure is defined by 
\begin{equation*}
R_u(x) = \frac{ - x U''(x)}{U'(x)},
\end{equation*}
so an equivalent way of stating condition \eqref{eq:Levy_varPr3} would be to require the  $R_u(x) \geq 1 $.
\end{remark}

We can use a concavity argument from the derivatives to get some form of uniqueness. A similar argument occurs in \cite{Kohatsu-Higa2006}, where it is proven that local maximums are unique in the case of logarithmic utility under some restriciton on admissible controls. In our case we have the following result:


\begin{theorem}
Suppose $A$ is a convex set in $\mathcal{A}_{\invFF}$ such that all $\pi\in A$ are bounded. If for all $\pi \in A$ assumptions \aui,  $A_{d^2}$ with $|y| < \delta_{\pi,\beta}$, $\beta \in \mathcal{A}_{\invFF}$ bounded, and \eqref{eq:Levy_varPr3} is satisfied, then there can at most be one local maximum in $A$.
\label{theorem:uniqueness}
\end{theorem}

\begin{proof}
Suppose $\pi_1, \pi_2 \in A$ are two local maximums. Let $\pi_2 - \pi_1 = \bar{\pi}$. Since A is convex, we have $\pi_1 +y\bar{\pi}\in A$ for $y\in [0,1]$. We note that 
\begin{equation*}
\frac{d}{dy} \mathbb{E} \big[ U \big( X_{\pi_1+y \bar{\pi}}(T) \big) \big]_{|y=a} = \frac{d}{d\zeta} \mathbb{E} \big[ U \big( X_{(\pi_1+a \bar{\pi}) + \zeta\bar{\pi}}(T) \big) \big]_{|\zeta=0} \quad \text{for } a\in [0,1].
\end{equation*}
Indeed \aui and $A_{d^2}$ hold for $(\pi_1+a \bar{\pi})$ as it is an element of $A$. In particular we can apply Lemma \ref{lemma:concavity} to conclude that $\frac{d}{dy} \mathbb{E} \big[ U \big( X_{\pi_1+y \bar{\pi}}(T) \big) \big]$
is strictly monotone for $y \in [0,1]$.


We show that there cannot exist two local maximums by contradiction. Consider 
\begin{align}
\frac{d}{dy} \mathbb{E} \big[ U \big( X_{\pi_1+y \bar{\pi}}(T) \big) \big]_{|y=1} = \frac{d}{d\zeta} \mathbb{E} \big[ U \big( X_{\pi_1+ \bar{\pi} + \zeta\bar{\pi}}(T) \big) \big]_{|\zeta=0} = \frac{d}{d\zeta} \mathbb{E} \big[ U \big( X_{\pi_2 + \zeta\bar{\pi}}(T) \big) \big]_{|\zeta=0} = 0
\label{eq:contradiction}
\end{align}
since $\pi_1+\bar{\pi} =\pi_2$, and $\pi_2$ is a local maximum. On the other hand, we also have that $\pi$ is a local maximum, hence
\begin{equation*}
\frac{d}{dy} \mathbb{E} \big[ U \big( X_{\pi_1+y \bar{\pi}}(T) \big) \big]_{|y=0}  = 0.
\end{equation*}
Consequently $\frac{d}{dy} \mathbb{E} \big[ U \big( X_{\pi_1+y \bar{\pi}}(T) \big) \big]$ is strictly monotone and zero at two different points, which is absurd.
\end{proof}

\subsection{Some examples with logarithmic utility}

We concentrate on the logarithmic utility to reduce computation and highlight some interesting aspects of the analysis. Note that if $U(x) = \ln(x)$ then $F_\pi(T) = 1$ in \eqref{eq:Fpi}. 
By application of Theorem \ref{theorem:localMax} the following equation plays a crucial role:
\begin{align}
0 =&\; \mathbb{E} \big[ F_{\pi}(T) \big( M_{\pi}(s) - M_{\pi}(t) \big) \Big| \invF_t \Big] \nonumber \\
=&\; \mathbb{E} \Big[  M_{\pi}(s) - M_{\pi}(t) \big| \invF_t \Big] \nonumber \\
=&\; \mathbb{E} \Big[ \int\limits_{t\wedge \tau} ^{s\wedge\tau} \big[ \mu(r)-\rho(r)-\sigma^2(r)\pi(r) - \int\limits_{\mathbb{R}_0} \frac{\pi(r) \theta^2(r,z)} {1+\pi(r)\theta(r,z) } \nu(dz) \big] dr \nonumber \\
&+ \int\limits_{t\wedge\tau}^{s\wedge\tau} \sigma(r) d^-W (r) + \int\limits_{t\wedge\tau}^{s\wedge\tau} \int\limits_{\mathbb{R}_0} \frac{\theta(r,z)}{1+\pi(r)\theta(r,z)} \tilde{N}(d^- r,dz) \nonumber \\
&+ \int\limits_{t\wedge\tau}^{s\wedge\tau} \frac{\kappa(r)}{1+\kappa(r)\pi(r)} dH(r) \Big| \invF_t \Big] \quad\quad s\geq t.
\label{eq:log_problem}
\end{align}

\begin{example}

Assume that $H$ is independent of $W$ and $N$ and that all the coefficients are $\mathbb{F}$-adapted, as in classical market modeling. Further we assume that $\Lambda(ds) = \lambda(s) ds$, for some positive stochastic process $\lambda$. 

We consider the case of an investor having access to an information flow $\invFF$ with $\invF_t \subseteq \mathcal{F}_t \vee \mathcal{F}^H_t$, for all $t\in[0,T]$. We call this a \emph{case of partial information}. The expectation of the forward integrals in \eqref{eq:log_problem} are zero in this setup, so the equation can be written 
\begin{align*}
0 =&\;  \mathbb{E} \Big[ \int\limits_{t\wedge\tau}^{s\wedge\tau} \Big[ \mu(r)-\rho(r)-\sigma^2(r)\pi(r) - \int_{\mathbb{R}_0} \frac{\pi(r) \theta^2(r,z)} {1+\pi(r)\theta(r,z) } \nu(dz) \Big] dr \nonumber \\
&+ \int\limits_{t\wedge\tau}^{s\wedge\tau} \frac{\kappa(r)}{1+\kappa(r)\pi(r)} dH(r) \Big| \invF_t \Big], \quad s\geq t.
\end{align*}
Dividing by $(s-t)$ and letting $s \to t$, we find that the locally optimal $\pi(t)$ in this case must satisfy 
\begin{align}
0 =&\; \mathbf{1}_{\{\tau>t\}} \mathbb{E} \Big[  \mu(t)-\rho(t)-\sigma^2(t)\pi(t) \nonumber \\
&- \int_{\mathbb{R}_0} \frac{\pi(r) \theta^2(r,z)} {1+\pi(r)\theta(r,z) } \nu(dz) + \frac{\kappa(t)}{1+\kappa(t)\pi(t)} \lambda(t) \Big| \invF_t \Big].
\label{eq:partial2}
\end{align}
For illustration, assume $\theta= 0$. Then \eqref{eq:partial2} yields a polynomial equation in $\pi(t)$ of degree 2:
\begin{align*}
0 =&\; \ind_{\tau>t} \Big( \E \Big[ \mu(t) -\rho(t) + \kappa(t)  \lambda(t) \Big| \invF_t \Big] + \pi(t) \E \Big[ \mu(t)\kappa(t) -\rho(t)\kappa(t) -\sigma^2(t)  \lambda(t) \Big| \invF_t \Big]\\
&-\pi^2(t) \E \Big[ \sigma^2(t) \kappa(t)  \lambda(t) \Big| \invF_t \Big]\Big)
\end{align*}
\end{example}

\begin{example}
Assume that $\invF_t = \mathcal{F}_t \vee \mathcal{F}^H_t$ and that $H$ contains no anticipating information on $\mathbb{F}$. In this case we say that the investor has \emph{full information}. If $\Lambda(ds) = \lambda(s) ds$ and the coefficients $\mu$, $\sigma$, $\theta$, $\kappa$ are adapted to $\invFF$, equation \eqref{eq:partial2} reduces to 
\begin{align}
0 =&\; \mathbf{1}_{\{\tau>t\}} \Big( \mu(t)-\rho(t)-\sigma^2(t)\pi(t) \nonumber \\
&- \int_{\mathbb{R}_0} \frac{\pi(r) \theta^2(r,z)} {1+\pi(r)\theta(r,z) } \nu(dz) + \frac{\kappa(t)}{1+\kappa(t)\pi(t)} \lambda(t) \Big)
\label{eq:full}
\end{align}
If we assume $\theta=0$, the explicit solution of \eqref{eq:full} is given by 
\begin{equation}
\pi = \frac{1}{2\kappa}\Bigg( \frac{\kappa(\mu-\rho)}{\sigma^2} - 1  + \sqrt{\Big(1-\frac{\kappa(\mu-\rho)}{\sigma^2}\Big)^2 + 4\kappa \Big(\frac{\mu-\rho+\lambda\kappa}{\sigma^2}\Big) } \; \Bigg),
\label{eq:full2}
\end{equation}
where we used \eqref{eq:pi_limitation} to exclude one of the two solutions of the quadratic eqauation.
Remark that equation \eqref{eq:full} gives us Merton ratio when $\kappa= \theta=0$.

\begin{figure}
\centering
\subfloat[Default risk \emph{not} compensated by higher drift]{\label{fig:logUtilityNoComp}\includegraphics{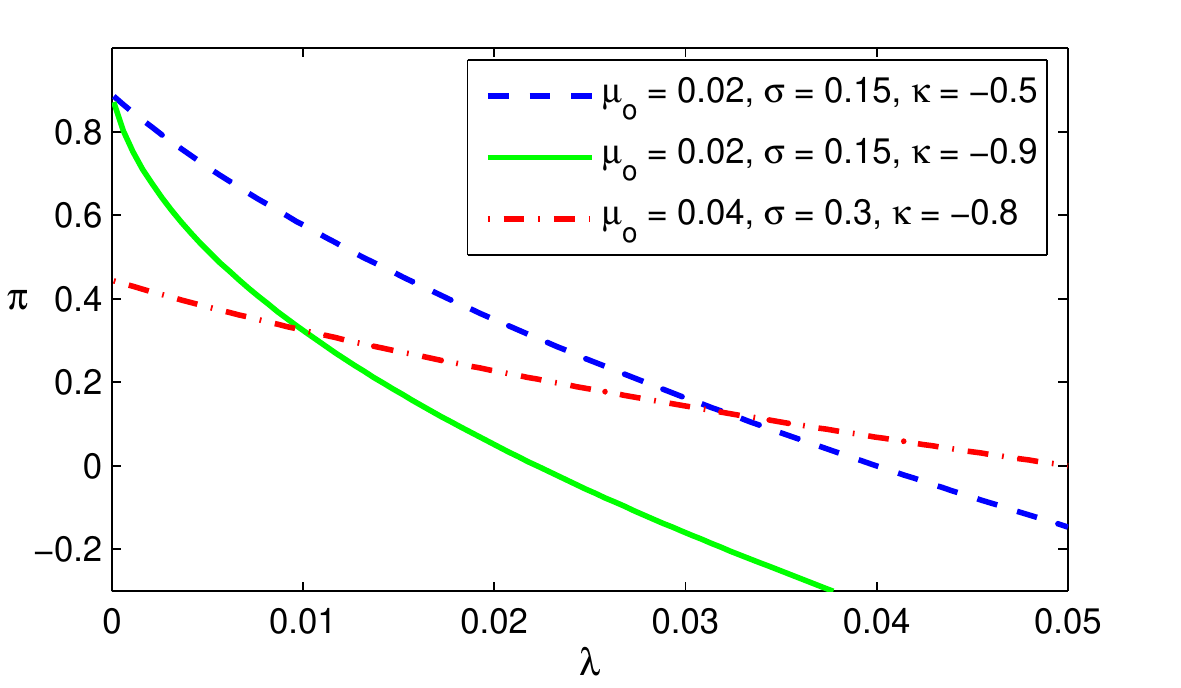}} \\  
\subfloat[Default risk compensated by higher drift]{\label{fig:logUtilityComp}\includegraphics{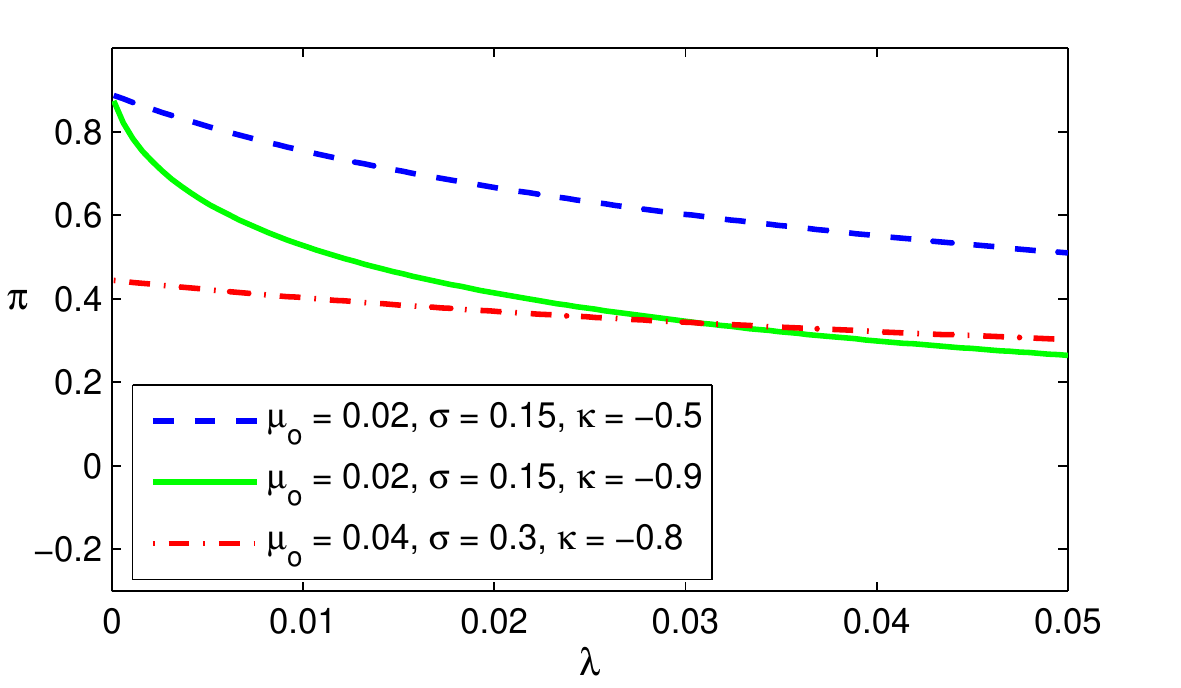}}
\caption{Optimal investment $\pi$ as a function of $\lambda$.}             
\label{fig:logutility}
\end{figure}


Two explicit examples with full information can be found in the figures. In Figure \ref{fig:logUtilityNoComp} the stock price is modeled as
\begin{equation}
dS_1(t) =  S(t-) \Big( \mu_o dt + \sigma dW(t) + \kappa dH(t)\Big), 
\label{eq:stockprice_simple}
\end{equation}
with $\mu_o, \sigma$, $\kappa$ fixed, and $\rho=0$. We see that with higher default risk the agent invests less and the asset is also shorted when the overall return becomes negative at the point $\lambda\kappa = - \mu_o$.

In Figure \ref{fig:logUtilityComp} the stock price is modeled as
\begin{equation}
dS_1(t) =  S(t-)\Big( \big(\mu_o -\lambda\kappa \big) dt + \sigma dW(t) + \kappa dH(t)\Big) 
\label{eq:stockprice_simple_comp}
\end{equation}
The assumptions in \eqref{eq:stockprice_simple_comp} are similar to \eqref{eq:stockprice_simple}. But with the term $-\lambda\kappa$ in the drift, the expected return of the asset is invariant to the value of $\lambda$. So the agent invests less due to risk aversion and not due to changes in the asset returns.
\end{example}


Next we explicitly detail how after-default and/or multiple defaults are easily treated in our framework.
\begin{example}
Here we discuss a model with default time $\zeta$. After default the asset has a recovery process with different dynamics than before default. We assume that it is possible to invest both before and after default and set $\tau = T$. Set $H(t) = \ind_{\{\zeta\leq t\}}(t)$ and
\begin{align*}
\mu(s) &= \mu_1 \ind_{\{H(s-) = 0 \}}(s) + \mu_2 \ind_{\{H(s-) > 0 \} }(s) \\
\theta(s,z) &= z \theta_1 \ind_{\{H(s-) = 0 \}}(s) + z \theta_2 \ind_{\{H(s-) > 0 \} }(s) \\
\sigma(s) &=  \sigma_1 \ind_{\{ H(s-) = 0 \}}(s) + \sigma_2 \ind_{\{ H(s-) > 0 \}}(s) 
\end{align*}
where $\mu_1, \mu_2, \theta_1, \theta_2, \sigma_1, \sigma_2 \in \mathbb{R}$. Here $\theta_1$ and $\sigma_1$ are the coefficients of the noises of the asset dynamics pre-default and $\theta_2$ and $\sigma_2$ the coefficients after default while $\mu_1$ and $\mu_2$ are the drift coefficients before and after default respectively. 

In the case of full information as above, $\invF_t = \mathcal{F}_t \vee \mathcal{F}^H_t$, the optimization scheme seperates into pre-default and after default. The optimal portfolio $\pi$ satisfies:

\begin{align*}
0 =&\; \mathbf{1}_{\{\zeta>t\}} \Big( \mu_1-\rho(t)-\sigma_1^2(t)\pi(t) \nonumber \\
&- \int_{\mathbb{R}_0} \frac{\pi(r) z\theta_1^2} {1+\pi(r)z \theta_1 } \nu(dz) + \frac{\kappa(t)}{1+\kappa(t)\pi(t)} \lambda(t) \Big) \\
&+ \mathbf{1}_{\{\zeta\leq t\}} \Big( \mu_2-\rho(t)-\sigma_2^2(t)\pi(t) \nonumber \\
&- \int_{\mathbb{R}_0} \frac{\pi(r) z\theta_2^2} {1+\pi(r)z \theta_2 } \nu(dz)  \Big).
\end{align*}
\end{example}

\bigskip

The cases of anticipating information, \ie $\invF_t \supseteq \mathcal{F}_t \vee \mathcal{F}^H_t$, are more subtle than partial or full information, with various approaches being possible depending on the specific conditions. The main challenge with anticipating information is to evaluate the terms $\mathbb{E} \big[ \int_{t\wedge\tau}^{s\wedge\tau} \sigma(r) d^-W (r) \big| \invF_t \big]$ and $\mathbb{E} \big[ \int_{t\wedge\tau}^{s\wedge\tau} \int_{\mathbb{R}_0} \frac{\theta(r,z)}{1+\pi(r)\theta(r,z)} \tilde{N}(dz,d^-r) \big| \invF_t \big]$ in \eqref{eq:log_problem}.

One possible way to compute the expectations of the forward integrals above is to exploit Malliavin calculus, see \cite[Chapter 8 and Chapter 15]{DiNunno2009} for the theoretical framework. However we must stress that this general approach cannot always be taken here. In fact the application of Malliavin calculus requires that the integrands are measurable with respect to $\mathcal{F}_T$, which is not, in general, the case when considering default risk. 

See also \cite{DiNunno2006} on how the $\tilde{N(}d^-t,dz)$ integral can be evaluated using predictable compensators of the measure with respect to $\invFF$ and \cite{Biagini2005,Leon2003,Kohatsu-Higa2006} for other examples on the $d^-W$ integral in insider models without default risk.

\bigskip
Hereafter we show an example where the process $H$ contains anticipating information on the jumps of $N$. 
In this example the process $\int_0^t \int_\Rr z \tilde{N}(ds,dz)$, $t\in [0,T]$, is martingale under $(\Prob, \mathbb{F})$ but not under $(\Prob, \invFF)$. While we claim no particular market model related to this, we show how the suggested framework enables solutions for optimization problems with anticipating information. In particular the process $H$ generalizes the optimization problem not only because of it's precense, but also because it adds knowledge on the other noises driving $S_1$. Nevertheless a solution is obtained from Theorem \ref{theorem:localMax}.


\begin{example}
Assume $\nu(dz) = \gamma \ind_{\{1\}}(dz)$, \ie the Poisson random measure $N(dt,dz)$ is actually $dN(t)$ where $N(t),\, t\geq 0$ is a Poisson process with intensity $\gamma$. Define $\tau = \inf_{t} \{ N(t+\epsilon) - N(t-\epsilon) \geq 2 \}$, $\epsilon >0$, and $H(t) = \ind_{\{ \tau \leq t\}}(t)$. Thus $H$ is independent of $W$ but \emph{is} dependent on $N$. Note that $\tau$ contains anticipating information with respect to $\mathbb{F}$ since $(\tau > t)$ implies $N(t+\epsilon) - N(t-\epsilon) < 2$. Set $\invF_t = \mathcal{F}_t \vee \mathcal{F}_t^H$.

Our ad hoc interpretation is that too many bad events (represented by $N$) in a limited time span ($2 \epsilon$) will cause the firm to default (with a loss $\kappa H$ and the asset is no longer tradeable $\tau$).

We assume $\rho$, $\mu$, $\sigma$, $\theta$ and $\kappa$ are constants. Starting from \eqref{eq:log_problem}
\begin{align*}
0 =&\; \mathbb{E} \Big[ \int\limits_{t\wedge\tau}^{s\wedge\tau} \big[ \mu-\rho-\sigma^2\pi(r) - \frac{\pi(r) \theta^2} {1+\pi(r)\theta^2} \gamma \big] dr 
+ \int\limits_{t\wedge\tau}^{s\wedge\tau} \sigma(r) d^-W (r) \\
&+ \int\limits_{t\wedge\tau}^{s\wedge\tau} \frac{\theta}{1+\pi(r)\theta} \big( N(dr)-\gamma dr \big) + \int\limits_{t\wedge\tau}^{s\wedge\tau} \frac{\kappa}{1+\kappa\pi(r)} dH(r) \Big| \invF_t \Big].
\end{align*}
Computing the predictable compensators of $H$ and $N$ (sketched below), and dividing by $(t-s)$ we find that the optimal $\pi$ is a solution of
\begin{align*}
0=&\; \mu-\sigma^2 \pi(r) - \frac{\pi(r) \theta^2} {1+\pi(r)\theta^2} \gamma + \frac{\theta(r,z)}{1+\pi(r)\theta} \frac{\gamma}{1+\gamma \epsilon} \ind_{\{N(t)-N(t-\epsilon) = 0\}} \\
&+\ind_{\{N(t)-N(t-\epsilon) = 0\}} \frac{\gamma^2 \epsilon}{1+\gamma \epsilon} + \ind_{\{N(t)-N(t-\epsilon) = 1\}}\gamma.
\end{align*}
 

%

To compute the $\invFF$-predictable compensators of $H$ and $N$ we investigate the intensities (see, e.g. \cite[Section 3.2]{Bremaud1981}) on the set $\{t<\tau\}$
\begin{align*}
\lambda^N_t :=\lim_{\Delta t \to 0^+} \frac{1}{\Delta t} \E \big[ N(t+\Delta t) - N(t)  \big| \invF_t,t<\tau \big], \\
\lambda_t := \lim_{\Delta t \to 0^+} \frac{1}{\Delta t} \E \big[(H(t+\Delta t) - H(t) \big| \invF_t,t<\tau \big].
\end{align*}
The $\invFF$-predictable compensators of $H$ and $N$ are then given by $\Lambda(t) = \int_0^t \lambda(s) ds$ and $\Lambda^N (t) = \int_0^t \lambda^N(s) ds$. We consider the case of $N$, the computations for $H$ are similar. First note that 
\begin{equation*}
\lim_{\Delta t \to 0^+} \frac{1}{\Delta t} \Prob \big( N(t+\Delta t) - N(t) > 1 \big| \invF_t,t \leq \tau \big) = 0.
\end{equation*}
Recall that $(\tau > t)$ implies $N(t+\epsilon) - N(t-\epsilon) < 2$ and that $(\tau > t, N(t)-N(t-\epsilon) = 1 )$ implies $N(t+\epsilon) - N(t) = 0$. For $\Delta t < \epsilon$, 
\begin{align*}
\Prob \Big(& N(t+\Delta t) - N(t) =1  \Big| \invF_t,t\leq \tau \Big) = \\
& \ind_{\{ N(t)- N(t-\epsilon) = 0 \}} \Prob \Big( N(t+\Delta t) - N(t) = 1 \Big| N(t)- N(t-\epsilon) = 0, t \leq \tau \Big) \\
&+  \ind_{\{ N(t)- N(t-\epsilon) = 1 \}} \Prob \Big( N(t+\Delta t) - N(t) = 1 \Big| N(t)- N(t-\epsilon) = 1, t\leq \tau \Big) .
\end{align*}
%
We have
\begin{align*}
\Prob \Big( N(t+\Delta t) - N(t) = 1 \Big|& t\leq \tau ,N(t)-N(t-\epsilon)=0 \Big)  \\
=&\; \Prob\Big( N(t+\Delta t)-N(t) = 1 \Big| N(t+\epsilon)-N(t) < 2 \Big) \\ 
=&\; \frac{\Prob\Big( N(t+\Delta t)-N(t) = 1 ,\; N(t+\epsilon )-N(t+\Delta t) = 0 \Big) }{\Prob\Big( N(t+\epsilon )-N(t ) <2  \Big)} \\
=&\; \frac{\gamma \Delta t e^{-\gamma \Delta t} e^{-\gamma(\epsilon-\Delta t)}}{e^{-\gamma \epsilon} + \gamma \epsilon e^{-\gamma \epsilon}}
\end{align*}
and
\begin{equation*}
\Prob \Big( N(t+\Delta t) - N(t) = 1 \Big| t \leq \tau ,N(t)-N(t-\epsilon)=1 \Big) = 0.
\end{equation*}
Thus
\begin{equation*}
\lambda^N(t)=   \frac{\gamma}{1+\gamma \epsilon} \ind_{\{N(t)-N(t-\epsilon) = 0\}} \quad \text{for } t\leq \tau.
\end{equation*}
Similarly we find
\begin{equation*}
\lambda_t = \ind_{\{N(t)-N(t-\epsilon) = 0\}} \frac{\gamma^2 \epsilon}{1+\gamma \epsilon} + \ind_{\{N(t)-N(t-\epsilon) = 1\}}\gamma \quad \text{for } t\leq \tau.
\end{equation*}

\end{example}

\section{On the driving processes as semi-martingales}

The results of Theorem \ref{theorem:localMax} take a more specific form when $M_\pi$ is $\invFF$-adapted. This will be our standing assumption throughout the section, implying that $\mathcal{F}_t, \mathcal{F}_t^H \subset \invF_t$ for all $t\in[0,T]$ and that the integrands $\rho, \mu, \sigma, \theta$ and $\kappa$ are $\invFF$-adapted.

\begin{theorem} 
Suppose that $M_\pi$ is $\invFF$-adapted and that for $\pi \in \mathcal{A}_{\invFF}$ assumption \aui holds.
\begin{enumerate}
\renewcommand{\theenumi}{\roman{enumi})}
\renewcommand{\labelenumi}{\theenumi}
\item If $\pi$ is a local maximum, then $M_\pi(t)$, $t\in[0,T]$, is a martingale under $(\mathbb{Q}_{\pi},\invFF )$. 
\label{list:martingale1}
\item If $\pi$ is a local maximum, then the stochastic process
\begin{equation*}
\hat{M}_\pi(t) = M_\pi(t) - \int\limits_0^t \frac{1}{Z(s)}d[M_\pi,Z](s), \quad t\in[0,T],
\end{equation*}
is a martingale under $(\mathbb{P},\invFF )$. Here, we have set 
\begin{equation*}
Z(t) := \mathbb{E}_{\mathbb{Q}_\pi} \Big[ \frac{d\mathbb{P}}{d\mathbb{Q}_\pi} \big| \invF_t \Big] = \Big( \mathbb{E} \big[ F_\pi(T) \big| \invF_t \big] \Big)^{-1}.
\end{equation*}
\label{list:martingale2}
\end{enumerate}
%
Assume that the mapping $y \to \mathbb{E} \big[ U(X_{\pi+y\beta}(T) \big]$ is concave for all bounded controls $\beta \in \mathcal{A}_{\invFF}$. Then we also have the converse conclusions
\begin{enumerate}
\renewcommand{\theenumi}{\roman{enumi})}
\renewcommand{\labelenumi}{\theenumi}
\setcounter{enumi}{2}
\item If $M_\pi$ is a martingale under $(\mathbb{Q}_{\pi},\invFF)$, then $\pi$ is a local maximum.
\label{list:martingale3}
\item If the stochastic process
\begin{equation*}
\hat{M}_\pi(t) = M_\pi(t) - \int\limits_0^t \frac{1}{Z(s)} d[M_\pi,Z](s), \quad t\in[0,T],
\end{equation*}
is a martingale under $( \mathbb{P},\invFF)$, then $\pi$ is a local maximum.
\label{list:martingale4}
\end{enumerate}
\label{theorem:martingale}
\end{theorem}

\begin{proof} ~\newline
\noindent \textbf{Part \ref{list:martingale1}} Being $M_\pi$ $\invFF$-adapted, it is a $( \mathbb{P},\invFF)$-martingale.

\smallskip \noindent \textbf{Part \ref{list:martingale2}} is obtained by application of the Girsanov theorem (see in particular \cite[Part III, Theorem 39]{Protter2005}).

\smallskip \noindent \textbf{Part \ref{list:martingale3}} is a direct application of Theorem \ref{theorem:localMax}.

\smallskip \noindent \textbf{Part \ref{list:martingale4}} is again an application of the Girsanov theorem.
%
\end{proof}
The existence of a local maximum also has other implications. 

\begin{theorem}
If a local maximum exists, $M_\pi$ is $\invFF$-adapted and \aui holds, then $W$ and $\int_0^{t\wedge \tau} \int_{\Rr} \theta(s,z) \tilde{N}(ds,dz)$, $t\in [0,T]$, are semi-martingales under $(\mathbb{P},\invFF)$.
\label{theorem:semi-martingale}
\end{theorem}
%
%
\begin{proof}
Assume a local maximum $\pi \in \mathcal{A}_{\invFF}$ exists. By Theorem \ref{theorem:localMax} this implies that $M_\pi$ is a $(\mathbb{Q}_{\pi},\invFF)$-martingale. Define
\begin{equation*}
M_\pi^H (t) := \int\limits_0^{t\wedge \tau} \frac{\kappa(s)}{1+\pi(s) \kappa(s)} dH(s) - \int\limits_0^{t\wedge \tau}  \frac{\kappa(s)}{1+\pi(s) \kappa(s)} \Lambda^{\mathbb{Q}_\pi} (ds) , \quad t\in[0,T], 
\end{equation*}
where $\Lambda^{\mathbb{Q}_\pi}$ is the $(\mathbb{Q}_{\pi},\invFF)$-predictable compensator of $H$. Note that $M_\pi^H$ is a $(\mathbb{Q}_{\pi},\invFF)$-martingale and thus $M_\pi -M_\pi^H$ is also a $(\mathbb{Q}_{\pi},\invFF)$-martingale. We can (uniquely) decompose $M_\pi -M_\pi^H$ into a continuous martingale and pure jump martingale \cite[Theorem 1.4.18]{Jacod2003} which we denote by $M_\pi^W$ and $M_\pi^N$ respectively. By the definition of $M_\pi$ it follows that we can write
\begin{align*}
M_\pi^W(t) &:= \int\limits_0^{t\wedge \tau} \sigma(s) d^-W(s) + \int\limits_0^{t\wedge \tau} \sigma(s) a(s) ds , \quad t \in [0,T], \\
M_\pi^N(t) &:=  \int\limits_0^{t\wedge \tau} \int\limits_\Rr \frac{\theta(s,z)}{1+\pi(s) \theta(s,z)} \tilde{N}(d^-s,dz) + \int\limits_0^{t\wedge \tau} \gamma(s) ds ,\quad t\in[0,T], \\
\end{align*}
where 
$a$ and $\gamma$ is such that 
\begin{align}
\int\limits_0^{t\wedge \tau} \sigma(s) a(s) ds + \int\limits_0^{t\wedge \tau} \gamma(s) ds = \int\limits_0^{t\wedge \tau}  \mu(s) - \rho(s) -\sigma^2(s) \pi(s) \nonumber \\
- \int\limits_\Rr \frac{\theta^2(s,z)\pi(s)}{1+\pi\theta(s,z)}  \nu(dz) ds + \int\limits_0^{t\wedge \tau} \frac{\kappa(s)}{1+\kappa(s)\pi(s)} \Lambda^{\mathbb{Q}_\pi}(ds).
\label{eq:a_plus_gamma} 
\end{align}
%
The right hand side of \eqref{eq:a_plus_gamma} have a finite $\Prob$ expectation by the assumptions \eqref{eq:finite_default_integrals} and Definition \ref{definition:Allowable_controls1} so that $\int_0^t \sigma(s) a(s) ds$ and $\int_0^t \gamma(s) ds$ are processes of finite variation.

As in Theorem \ref{theorem:martingale}, $M_\pi^W - \int_0^{t\wedge \tau} \frac{1}{Z(s)} d[M_\pi^W,Z]$ is a $(\mathbb{P},\invFF)$-martingale. We note that $[M_\pi^W,Z]$ is absolutely continouos with respect to Lebesgue by the Kunita-Watanabe-inequality (see for instance \cite[Theorem 25]{Protter2005}) since the quadratic variation of $M_\pi^W$ is absolutely continouos with respect to Lebesgue. 
Thus the quadratic variation of $\int_0^{t\wedge \tau} \frac{1}{\sigma(s)}  M_\pi^W (ds)$ is $t$, making $\frac{1}{\sigma}  M_\pi^W $ a $(\mathbb{P},\invFF)$-Brownian motion. Hence $W$ has the $(\mathbb{P},\invFF)$ semi-martingale decomposition
\begin{equation*}
W(t) = \tilde{W}(t) + \int\limits_0^{t\wedge \tau} a(s) ds - \int\limits_0^{t\wedge \tau} \frac{1}{Z(s) \sigma(s) } [M_\pi^W,Z] ds 
\end{equation*}
where $\tilde{W}$ is a $(\mathbb{P},\invFF)$-Brownian motion.

Similarly, 
\begin{align}
& M_\pi^N(t) - \int\limits_0^{t\wedge \tau} \frac{1}{Z(s)} d[M_\pi^N,Z](s) \nonumber \\
& = \int\limits_0^{t\wedge \tau} \int\limits_\Rr \frac{\theta(s,z)}{1+\theta(s,z) \pi(s)} \tilde{N}(d^-s,dz)+\int\limits_0^{t\wedge \tau} \gamma(s) ds - \int_0^t \frac{1}{Z(s)} d[M_\pi^N,Z](s), \quad t\in[0,T],
\label{eq:MpiN}
\end{align}
is a $(\mathbb{P},\invFF)$-martingale. We have that
\begin{equation*}
\int\limits_0^{t\wedge \tau} \frac{1}{Z(s)} d[M_\pi^N,Z](s) = \int\limits_0^{t\wedge \tau} \frac{ \E\big[ U\big(X_{\pi} (T) \big) X_{\pi}(T) \big| \invF_s \big]}{ \E\big[ U\big(X_\pi (T)\big) X_\pi (T) \big]   } d[ M_{\pi}^N,Z](s), \quad t\in[0,T],
\end{equation*}
is a $\invFF$-adapted process of finite variation and thus a $(\Prob,\invFF)$ semi-martingale (recall that $[ M_{\pi}^N,Z]$ is of finite variation \cite[p. 67]{Protter2005}). Since also $M_\pi^N$ and $\int_0^t \gamma(s) ds$  are $(\Prob,\invFF)$ semi-martingales we must have that 
\begin{equation*}
\int\limits_0^{t\wedge \tau} \int\limits_\Rr \frac{\theta(s,z)}{1+\theta(s,z) \pi(s)} \tilde{N}(d^-s,dz), \quad t\in[0,T]
\end{equation*}
is a $(\Prob,\invFF)$ semi-martingale by \eqref{eq:MpiN}. Since the $1+\theta \pi$ is bounded away from zero (recall Definition \ref{definition:Allowable_controls1}) we must also have that $\int_0^{t\wedge \tau} \int_\Rr \theta(s,z) \tilde{N}(d^-s,dz)$ is a semi-martingale.
\end{proof}
Finally we do an analysis on the jumps of $H$.
\begin{theorem}
Assume that a local maximum exists, $N=0$, $M_\pi$ is $\invFF$-adapted and Assumption \aui holds. Then the jumps of $H$ are totally inaccessible stopping times (for the filtration $\invFF$) and the $(\mathbb{Q}_\pi,\mathbb{G})$-predictable compensators of $H$ is absolutely continuous with respect to Lebesgue.
\label{teorem:totally_inaccessible}
\end{theorem}
%
%
\begin{proof}
With the above assumptions, $M_\pi$ is a $(\mathbb{Q}_\pi,\mathbb{G})$-martingale by Theorem \ref{theorem:localMax}.
Denote 
\begin{align*}
A_1(t) &= M_\pi(t)  - \int\limits_0^t \frac{\kappa(s)}{1+\kappa(s)\pi(s)} H(ds) \\
A_2(t) &= \int\limits_0^t \frac{\kappa(s)}{1+\kappa(s)\pi(s)} H(ds), 
\end{align*}
for $t\in [0,T]$. Remark that $A_1+A_2 = M_\pi$ is a martingale. Since $A_2$ is discontinuous, $A_1$ must be the sum of the predictable compensator of $A_2$ and a martingale. Hence, since $A_1$ is continuous the compensator of $A_2$ is continuous. It immediately follows that the jump times of $H$ are totally inaccesible (see, e.g. \cite[Corollary 22.18]{Kallenberg1997}). 
\end{proof}


\clearpage
\bibliographystyle{siam}
\bibliography{referanser} 
\end{document}